\documentclass{article}

\usepackage[textwidth=14cm,textheight=21cm]{geometry}
\usepackage{graphics}
\usepackage{amsthm}

\def\n{\noindent}
\newtheorem{theorem}{Theorem}[section]

\newtheorem{lemma}[theorem]{Lemma}

\newtheorem{corollary}[theorem]{Corollary}
\newtheorem{observation}[theorem]{Observation}

\newtheorem{proposition}[theorem]{Proposition}

\usepackage{amssymb}
\usepackage{amsmath,amsfonts,amssymb}
\usepackage{fancyhdr}
\usepackage{soul,color}
\usepackage{graphicx,float}
\usepackage{tikz}
\usetikzlibrary{shapes,snakes}
\usetikzlibrary{arrows}
\usepackage{cite}
\usepackage{authblk}

\usepackage{abstract}

\begin{document}

\def\bubble#1#2)#3{
 \def\drawbubble##1##2{
   \expandafter\node\opts,draw=none]at(##1){##2};
 }
 \def\opts{[black}
 \def\defopts[##1](##2)##3{
   \def\opts{[##1}
   \drawbubble{##2}{##3}
 }
 \let\x=#1
 \if(\x
   \drawbubble{#2}{#3}
 \else\if[\x
   \defopts#1#2){#3}
 \else
   *** USE OF \backslash bubble DOES NOT MATCH ITS DEFINITION ***
 \fi\fi
}

\title{\bf Contracting Graphs to Split Graphs and Threshold Graphs}
\author{Leizhen Cai}
\author{Chengwei Guo}
\affil{\small Department of Computer Science and Engineering, The Chinese University of Hong~Kong, Hong~Kong~S.A.R., China}
\affil{\small \{lcai,cwguo\}@cse.cuhk.edu.hk}
\date{}
\maketitle

\begin{abstract}
\n{\bf Abstract.}
We study the parameterized complexity of {\sc Split Contraction} and {\sc Threshold Contraction}. In these problems we are given a graph $G$ and an integer $k$ and asked whether $G$ can be modified into a split graph or a threshold graph, respectively, by contracting at most $k$ edges. We present an FPT algorithm for {\sc Split Contraction}, and prove that {\sc Threshold Contraction} on split graphs, i.e., contracting an input split graph to a threshold graph, is FPT when parameterized by the number of contractions. To give a complete picture, we show that these two problems admit no polynomial kernels unless $NP\subseteq coNP/poly$.
\end{abstract}

\linespread{1}
\selectfont


\section{Introduction}

Graph modification problems constitute a fundamental and well-studied family of problems in algorithmic graph theory, many famous graph problems can be formulated as graph modification problems such as {\sc Clique}, {\sc Feedback Vertex Set}, and {\sc Minimum Fill-in}. A graph modification problem takes as input a graph $G$ and an integer $k$, and the question is whether $G$ can be modified to belong to a specified graph class, using at most $k$ operations of a certain specified type such as vertex deletion or edge deletion. The number $k$ of operations measures how close a graph is to such a specified class of graphs. Recently the study of modifying a graph by using operations of edge contraction has been initiated from the parameterized point of view, yielding several results for the {\sc $\Pi$-Contraction} problem.

\vskip 0.25cm
\begin{tabular}{l l}
\multicolumn{2}{l}{\sc $\Pi$-Contraction}\\
{\it Instance}: & Graph $G = (V,E)$, positive integer $k$.\\
{\it Question}: & Can we obtain a $\Pi$-graph (i.e. a graph belonging to class $\Pi$)\\
 & from $G$ by contracting at most $k$ edges?\\
{\it Parameter}: & $k$.\\
\end{tabular}
\vskip 0.25cm

A problem (with a particular parameter $k$) is \emph{fixed-parameter tractable} (FPT) if it can be solved in time $f(k)n^{O(1)}$ where $f(k)$ is a computable function depending only on $k$. Considering parameterization by the number of edge contractions used to modify graphs, the {\sc $\Pi$-Contraction} problem has been proved to be FPT when $\Pi$ is the class of bipartite graphs (Heggernes \emph{et al.}~\cite{heggernes11a}), the class of trees or paths (Heggernes \emph{et al.}~\cite{heggernes11c}), the class of planar graphs (Golovach \emph{et al.}~\cite{golovach12a}), the class of cliques (Cai \emph{et al.}~\cite{cai13} and Lokshtanov \emph{et al.}~\cite{lokshtanov13}). On the negative side, very recently two groups of authors (Cai \emph{et al.}~\cite{cai13} and Lokshtanov \emph{et al.}~\cite{lokshtanov13}) showed that {\sc Chordal Contraction} is \emph{fixed-parameter intractable}. It is then natural to ask whether the {\sc $\Pi$-contraction} problems are FPT for two well-known subclasses $\Pi$ of chordal graphs: split graphs and threshold graphs.

In this paper, we study the parameterized complexity of {\sc $\Pi$-Contraction} when $\Pi$ is the class of split graphs and when $\Pi$ is the class of threshold graphs. Their edge deletion versions known as {\sc Split Deletion} and {\sc Threshold Deletion}, asking whether an input graph can be modified into a split graph or a threshold graph, respectively, by deleting at most $k$ edges, are FPT and have polynomial kernels~\cite{cai96,guo07a,ghosh12}. The combination of split graphs, threshold graphs, and edge contractions has been studied in a closely related setting. Belmonte et al.~\cite{belmonte11a} showed that given a split graph $G$ and a threshold graph $H$, the problem of determining whether $G$ is contractible to $H$ is NP-complete, and it can be solved in polynomial time for fixed $|V(H)|$. Inspired by this, we consider the case that $H$ is an arbitrary threshold graph that is not a part of the input, and take as parameter the number of edge contractions instead of the size of the target graph $H$.\\

\n{\bf Our Contribution:} We show that {\sc Split Contraction} can be solved in $2^{O(k^2)}n^{O(1)}$ time. This result complements the FPT results of two other graph modification problems related to split graphs: {\sc Split Deletion} and {\sc Split Vertex Deletion}. Our algorithm starts by finding a large split subgraph and further is partitioned into two parts in terms of the clique size of this split subgraph. If the clique is large, we use a branch-and-search algorithm to enumerate edge contractions and reduce to the {\sc Clique Contraction} problem that is known to be FPT. Otherwise, there will be a large independent set in the input graph. We partition all vertices into bounded number of independent sets such that vertices in each set have the same neighbors, and then use reduction rules to reduce to a smaller graph. The ideas of reduction rules are applicable to obtain kernelization algorithms for other contraction problems such as {\sc Clique Contraction} and {\sc Biclique Contraction}.

We also obtain an $2^{O(k^2)}n^{O(1)}$ time algorithm for {\sc Threshold Contraction} when the input is a split graph. One motivation to study this problem is for considering an aspect of modification problems: contracting few edges to make all vertices in a graph obey some ordering (in this problem the ordering by neighborhood inclusion). We feel that it is of more interest to study the {\sc Threshold Contraction} problem for split graphs than for general graphs, since in the latter case the part of contracting input graphs to split graphs and the part of ordering vertices by edge contractions have lack of connection.

Furthermore, we prove that the above problems admit no polynomial kernels unless $NP\subseteq coNP/poly$ by polynomial-time reductions from {\sc Red-Blue Dominating Set} and {\sc One-Sided Dominating Set}, respectively, with polynomial bounds on the new parameter values. The incompressibility of {\sc Red-Blue Dominating Set} is provided by Dom \emph{et al.}~\cite{dom09}. We introduce {\sc One-Sided Dominating Set} and {\sc One-Sided Domatic Number} that were defined by Feige \emph{et al.}~\cite{feige02}, and show that these problems are incompressible. Such results might have application in obtaining kernelization lower bounds for further problems.


\section{Preliminaries}

\n{\bf Graphs:} We consider simple and undirected graphs $G=(V,E)$, where $V$ is the vertex set and $E$ is the edge set. Two vertices $u,v\in V$ are \emph{adjacent} iff $uv\in E$. A vertex $v$ is \emph{incident} with an edge $e$ iff $v\in e$, i.e., $v$ is an endpoint of $e$. The \emph{neighbor set} $N_G(v)$ of a vertex $v\in V$ in graph $G$ is the set of vertices that are adjacent to $v$ in $G$. We use $N_G[v]$ to denote the \emph{closed neighbor set} of $v$ in $G$ where $N_G[v]=N_G(v)\cup\{v\}$. For a set $X$ of vertices or edges in $G$, we use $G\setminus S$ or $G-X$ to denote the graph obtained by deleting $X$ from $G$. For a set of vertices $V'\subseteq V$, we denote by $E[V']$ the set of edges whose both endpoints are in $V'$. An \emph{$l$-path} (or $P_l$) is a path on $l$ vertices.

We study two hereditary graph classes. A graph $G$ is a \emph{split graph} if its vertex set can be partitioned into a clique $K$ and an independent set $I$, where $(K;I)$ is called a \emph{split partition} of $G$. The class of split graphs is characterized by a set of forbidden induced subgraphs: $\{2K_2,C_4,C_5\}$. A graph $G$ is a \emph{threshold graph} if there exist non-negative reals $r(v)$ for each $v\in V(G)$ and ``threshold'' $t$ such that for every vertex set $X\subseteq V(G)$, $X$ is an independent set iff $\Sigma_{v\in X}r(v)\leq t$. The class of threshold graphs is characterized by a set of forbidden induced subgraphs: $\{2K_2,C_4,P_4\}$.

\vskip 0.25cm
\n{\bf Edge Contraction:} The \emph{contraction} of edge $e=uv$ in $G$ removes $u$ and $v$ from $G$, and replaces them by a new vertex adjacent to precisely those vertices which were adjacent to at least one of $u$ or $v$. The resulting graph is denoted by $G/e$ or $G\cdot e$. For a set of edges $F\subseteq E(G)$, we write $G/F$ to denote the graph obtained from $G$ by sequentially contracting all edges from $F$.

For a graph $H$, if $H$ can be obtained from $G$ by a sequence of edge contractions, then $G$ is \emph{contractible} to $H$, or called \emph{$H$-contractible}. Let $V(H)=\{h_1,\cdots,h_l\}$. $G$ is $H$-contractible if $G$ has a so-called \emph{$H$-witness structure}: a partition of $V(G)$ into $l$ sets $W(h_1),\cdots,W(h_l)$, called \emph{witness sets}, such that each $W(h_i)$ induces a connected subgraph of $G$ and for any two $h_i,h_j\in V(H)$, there is an edge between $W(h_i)$ and $W(h_j)$ in $G$ iff $h_ih_j\in E(H)$. We obtain $H$ from $G$ by contracting vertices in each $W(h_i)$ into a single vertex.

\vskip 0.25cm
\n{\bf Parameterized Complexity:} A \emph{paramerized problem} $\mathcal{Q}$ is a subset of $\Sigma^*\times\mathbb{N}$ for some finite alphabet $\Sigma$. The second component is called the \emph{parameter}. The problem $\mathcal{Q}$ is \emph{fixed-parameter tractable} if it admits an algorithm deciding whether $(I,k)\in\mathcal{Q}$ in time $f(k)|I|^{O(1)}$, where $f$ is a computable function depending only on $k$.

A \emph{kernelization} of $\mathcal{Q}$ is a polynomial-time computable function that maps instance $(I,k)$ to another instance $(I',k')$ such that:
\begin{itemize}
\item $(I,k)\in\mathcal{Q}\Leftrightarrow (I',k')\in\mathcal{Q}$;
\item $|I'|,k'\leq g(k)$ for some computable function $g$.
\end{itemize}
If $g$ is a polynomial function then we say that $\mathcal{Q}$ admits a \emph{polynomial kernel}. A problem $\mathcal{Q}$ is \emph{incompressible} if it admits no polynomial kernel unless $NP\subseteq coNP/poly$.


\section{Contracting graphs to split graphs}

In this section, we consider the {\sc Split Contraction} problem: Given a graph $G$ and an integer $k$, can we obtain a split graph from $G$ by contracting at most $k$ edges?

We point out that this problem is NP-complete by reducing from another NP-complete edge contraction problem {\sc Clique Contraction}~\cite{cai13}: Given a graph $G$ and an integer $k$, can we modify $G$ into a clique by contracting at most $k$ edges? We construct a graph $G'$ from $G$ by adding an independent set of $k+2$ new vertices and making new vertices adjacent to all vertices in $G$. Observe that at least two of the new vertices are not involved in any edge contraction, then at least one of them belongs to the independent set of the resulting split graph, which implies that all the old vertices belong to the clique of the resulting graph. Thus $G'$ can be modified into a split graph using $k$ edge contractions iff there exists a set of $k$ edges in $G$ whose contraction makes $G$ into a clique.

\begin{theorem}\label{splitnpc}
{\sc Split Contraction} is NP-complete.
\end{theorem}
\vskip 0.25cm

We now present an FPT algorithm for {\sc Split Contraction} based on an $k^{O(k)}+O(m)$ time algorithm for {Clique Contraction}~\cite{cai13}.

Note that an $n$-vertex graph $G$ must contain an induced split subgraph of $(n-2k)$ vertices if $(G,k)$ is yes-instance of {\sc Split Contraction}, because $k$ edge contractions can affect at most $2k$ vertices. We start by finding an $(n-2k)$-vertex induced split subgraph $H$ in $2^{2k}n^{O(1)}$ time using a known algorithm for {\sc Split Vertex Deletion} (Ghosh \emph{et al.}~\cite{ghosh12}). Let $V_k=V(G)-V(H)$, and let $(K_H;I_H)$ be a split partition of $H$ where $K_H$ is a maximal clique and $I_H$ is an independent set. Here we first assume that $|K_H|>2k$ implying that at least one vertex in $K_H$ is not involved in any edge contraction, and will discuss the case for $|K_H|\leq 2k$ in the last part of the algorithm.

We branch out by contracting every possible set $E'\subseteq E[V_k]$ of at most $k$ edges and obtain the resulting instance $(G',k')$ where $G'=G/E'$ and $k'=k-|E'|$. For each resulting instance $(G',k')$, suppose that vertices $V_k$ are contracted to vertices $V'_k$ in $G'$.

\begin{proposition}
$(G,k)$ has a solution $S$ iff there exists a resulting instance $(G',k')$ such that $(G',k')$ has a solution $F=S-E'$ satisfying $F\cap E[V'_k]=\emptyset$.
\end{proposition}

Suppose that $(G',k')$ is a yes-instance and has a solution $F$. Thus the graph $G'/F$ is a split graph and has a split partition $(K_F;I_F)$. We further branch on at most $3^{|V'_k|}$ ways to find a partition $V'_k=R\cup K_p\cup I_p$ such that $R$ consists of exactly those vertices in $V'_k$ that are incident with some edges in $F$, $K_p\subseteq K_F$ induces a clique, and $I_p\subseteq I_F$ induces an independent set. See Fig. 1 for an illustration. It is easy to see that $|R|\leq k'$.

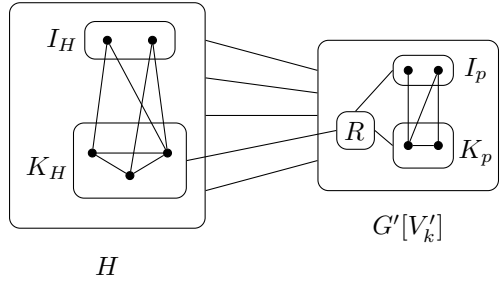
\begin{figure}[H]
\centering
\begin{tikzpicture}
  \tikzset{every node/.style={circle,draw,inner sep=0pt,minimum size=3pt}}
  \tikzstyle{cloud} = [draw, ellipse,fill=red!20, node distance=3cm, minimum height=2em]
  
  \node[draw,rectangle,rounded corners,text width=2.6cm,minimum height=3cm] at(2,3) (h){};
  \node[draw,rectangle,rounded corners,text width=1.2cm,minimum height=0.5cm] at(2.3,4) (ih){};
  \node[draw,rectangle,rounded corners,text width=1.5cm,minimum height=1cm] at(2.3,2.4) (kh){};
  \node[draw,rectangle,rounded corners,text width=2.4cm,minimum height=2cm] at(6,3) (vk){};
  \node[draw,rectangle,rounded corners,text width=0.5cm,minimum height=0.5cm] at(5.3,2.8) (r){$\;R\;$};
  \node[draw,rectangle,rounded corners,text width=0.8cm,minimum height=0.4cm] at(6.2,3.6) (ip){};
  \node[draw,rectangle,rounded corners,text width=0.8cm,minimum height=0.6cm] at(6.2,2.6) (kp){};
  \node[fill=black] at (2,4) (a1){};
  \node[fill=black] at (2.6,4) (a2){};
  \node[fill=black] at (1.8,2.5) (b1){};
  \node[fill=black] at (2.3,2.2) (b2){};
  \node[fill=black] at (2.8,2.5) (b3){};
  \draw(b1)--(b2)--(b3)--(b1);
  \draw(b1)--(a1)--(b3);
  \draw(b2)--(a2)--(b3);
  \node[fill=black] at (6,3.6) (c1){};
  \node[fill=black] at (6.4,3.6) (c2){};
  \node[fill=black] at (6,2.6) (d1){};
  \node[fill=black] at (6.4,2.6) (d2){};
  \draw(c1)--(d1)--(d2)--(c2)--(d1);
  \bubble(1.4,4){$I_H$}
  \bubble(1.2,2.3){$K_H$}
  \bubble(6.9,3.6){$I_p$}
  \bubble(6.9,2.5){$K_p$}
  \bubble(2,1){$H$}
  \bubble(6,1.5){$G'[V'_k]$}
  \begin{scope}
    \foreach \i in {-2,0,1,2}{
      \draw ([yshift=\i*0.5cm]h.east)--([yshift=\i*0.3cm]vk.west);
    }
  \end{scope}
  \draw(r.west)--(kh.east);\draw(r.north)--(ip.west);\draw(r.east)--(kp.west);

\end{tikzpicture}
\caption{An illustration of the structure of $G'$}
\end{figure}

It is clear that those vertices in $I_H$ that are adjacent to some vertices in $I_p$ must be in the clique $K_F$ of the target split graph, and other vertices in $I_H$ could be in the independent set $I_F$ after contractions. The following proposition states that almost all vertices in $K_F$ are finally in the clique of some target graph.

\begin{proposition}
If $(G',k')$ is a yes-instance, then it has a solution $F$ such that there is at most one vertex in $K_H$ that is finally in the independent set $I_F$.
\end{proposition}

\begin{proof}
For an arbitrary solution $F$ of $(G',k')$, if there are at least two vertices $a,b\in K_H$ that are finally in the independent set $I_F$, they must be contained in a same witness set $W_F$ because they are adjacent originally. Thus the number of such vertices is bounded by $k+1$, implying that there is a vertex $u\in K_H$ that is finally in the clique $K_F$ since $|K_H|>2k$.

By the definition of witness set we see that the induced subgraph $G'[W_F]$ has a spanning tree whose edges are entirely in $F$. Thus we can remove one edge from $F$ to separate the vertex $b$ from the witness set $W_F$, and add an edge $ub$ into $F$ which implies that $b$ is adjacent to all vertices in $K_F$ after contracting $F$. It is easy to see that the resulting set obtained from $F$ is also a solution of $(G',k')$ and $b$ is no longer in the independent set.
\end{proof}

Our algorithm further considers the following two cases. It outputs ``YES'' if either case outputs ``YES''.

\vskip 0.25cm
{\em {\bf Case 1.} There is no vertex in $K_H$ that is finally in $I_F$.}
Let $T_1$ be a subset of $I_H$ containing exactly those vertices that are adjacent to some vertices in $I_p$. It is clear that all vertices in $T_1$ are finally in $K_F$ after contractions. Therefore every vertex in $T_1$ must be involved in at least one edge contraction, and we have $|T_1|\leq k'$.

Remember that $R$ consists of exactly those vertices in $V'_k$ that are incident with some edges in $F$, which implies that vertices in $R$ are merged into $K_H\cup I_H$ after contracting some edges in $F$. It is easy to see that in order to contract $G'$ into a split graph, it is better to merge vertices $R$ into $K_H\cup T_1$ than into $I_H-T_1$. Thus we may assume that all vertices in $R$ are finally in $K_F$ after contractions. Our goal becomes to check whether $T_1\cup K_H\cup R\cup K_p$ induces a clique after contracting at most $k'$ edges in $G'$.

Since $|K_H|>2k$, there exists a vertex $u\in K_H$ that is not involved in any edge contraction. Obviously $u$ is adjacent to all vertices in $K_F$. We can obtain an edge set $F_1$ from $F$ by removing every edge of $F$ whose endpoints are both outside $K_H$, and replacing every edge $ab\in F$ that $a\in I_H-T_1,b\in K_H$ by an edge $ub$. Since contracting such edge $ab$ only affects the vertex $b$ which can be merged into the clique $K_F$ by contracting $ub$, it is clear that $F_1$ is also a solution set of $G'$ consisting of edges whose one endpoint is in $K_H$ and another endpoint is in $T_1\cup K_H\cup R\cup K_p$.

\begin{proposition}
$(G',k')$ has a solution set that is entirely contained in $G'[T_1\cup K_H\cup R\cup K_p]$ if it is a yes-instance for Case 1.
\end{proposition}

By the above result, we first find the set $T_1$ in polynomial time, and then apply the FPT algorithm for {\sc Clique Contraction}~\cite{cai13} to determine whether $G'[T_1\cup K_H\cup R\cup K_p]$ can be made into a clique by at most $k'$ edge contractions. If it outputs ``YES'', then $(G',k')$ is a yes-instance. The running time is bounded by $k'^{O(k')} +O(m)$.

\vskip 0.25cm
{\em {\bf Case 2.} There is exactly one vertex in $K_H$ that is finally in $I_F$.}
We check whether $(G',k')$ is a yes-instance with the assumption $|K_H\cap I_F|=\{w\}$ for every $w\in K_H$ that is not adjacent to any vertex in $I_p$.

Let $T_2$ be a subset of $I_H$ containing exactly those vertices that are adjacent to some vertices in $I_p\cup\{w\}$. Since each vertex in $T_2$ is finally in $K_F$ and thus is involved in some edge contraction, we have $|T_2|\leq k'$. Our goal is to check whether $T_2\cup (K_H-\{w\})\cup R\cup K_p$ induces a clique after contracting at most $k'$ edges in $G'$. Similar to Case 1, we have the following proposition.

\begin{proposition}
$(G',k')$ has a solution set that is entirely contained in $G'[T_2\cup (K_H-\{w\})\cup R\cup K_p]$ if it is a yes-instance for Case 2.
\end{proposition}

We also apply the $k'^{O(k')} +O(m)$ time algorithm for {\sc Clique Contraction}~\cite{cai13} to determine whether $G'[T_2\cup (K_H-\{w\})\cup R\cup K_p]$ can be made into a clique by contracting at most $k'$ edges. If it outputs ``YES'', then $(G',k')$ is a yes-instance.

\vskip 0.25cm
Combing Case 1 and Case 2, we can decide whether a resulting instance $(G',k')$ is a yes-instance in $k^{O(k)}+O(m)$ time when $|K_H|>2k$.\\

Furthermore, we deal with the remaining case: $|K_H|\leq 2k$. We partition the vertex set $V(G)$ into disjoint sets $X_1,\cdots,X_d$ such that each $X_i$ induces a maximal independent set satisfying that the vertices in $X_i$ have the same sets of neighbors in $G$. This procedure is equivalent to partitioning the complement graph $\overline{G}$ into critical cliques, which can be done in linear time~\cite{lin00,chen10a}. A \emph{critical clique} $K$ in a graph is a clique such that all vertices in $K$ have the same closed neighbor sets, and $K$ is maximal under this property. It has been proved that all vertices in a graph can be uniquely partitioned into groups such that each group induces a critical clique~\cite{lin00,chen10a}. We now use the following reduction rules to obtain a smaller instance:

\vskip 0.25cm
\n{\bf Rule 1.}\quad If $d>2^{4k}+4k$, then output ``NO''.

\vskip 0.25cm
\n{\bf Rule 2.}\quad If there are more than $2k+5$ vertices in $X_i$ for some $i$, then arbitrarily retain $2k+5$ vertices among them and remove others in $X_i$ from $G$.\\

It is clear that applying these reduction rules requires linear time. We now show the correctness of Rule 1 and Rule 2.

\vskip 0.25cm
\begin{lemma}
Rule 1 and Rule 2 are correct.
\end{lemma}

\begin{proof}
Since $|K_H|\leq 2k$, we have $|V_k\cup K_H|\leq 4k$ and $|I_H|\geq n-4k$. Note that vertices in $I_H$ have at most $2^{4k}$ different connection configurations to vertices $V_k\cup K_H$. Thus $I_H$ can be partitioned into at most $2^{4k}$ maximal independent sets such that vertices in each set have the same neighbors in $G$. Including vertices of $V_k\cup K_H$, the number $d$ is bounded by $2^{4k}+4k$ if $(G,k)$ is a yes-instance. Thus Rule 1 is correct.

Moreover, we prove that the input graph $G$ has a $k$-solution iff the graph $G^*$ obtained after one application of Rule 2 has a $k$-solution. Let $Y_i$ be the set of remaining vertices in $X_i$ for $i=1,\cdots, d$.

Suppose that $G$ has a solution $S\subseteq E(G)$. For every vertex $a$ that is incident with edges in $S$ and is removed after applying Rule 2, there exists $1\leq j\leq d$ such that $a\in X_j-Y_j$. Note that $|Y_j|=2k+5>2k$, there is another vertex $b\in Y_j$ that is not incident with any edge in $S$. We replace all edges $\{ac\in S:c\in V(G)\}$ by $\{bc:ac\in S\}$ in $S$ for every such $a$, and obtain a set $S'\subseteq E(G^*)$. It is clear that $G^*/S'$ is an induced subgraph of $G/S'$ because $S'\subseteq E(G^*)$. Since $N_G(a)=N_G(b)$ for every such $a$, it is easy to see that $G/S'$ is isomorphic to $G/S$ that is a split graph. Therefore $G^*/S'$ is a split graph.

Conversely, suppose that $G^*$ has a solution $S^*$, i.e., $G^*/S^*$ is a split graph. We claim that $G/S^*$ is also a split graph. Towards a contradiction, we assume that $G/S^*$ contains an induced subgraph $D$ isomorphic to one graph of $\{2K_2,C_4,C_5\}$. For every vertex $a$ that is contained in $V(D)$ and is removed after applying Rule 2, $a$ is in $X_j-Y_j$ for some $j$. Since $S^*\subseteq E(G^*)$, $a$ is not incident with any edge in $S^*$. Note that $|Y_j|=2k+5\geq 2k+|V(D)|$, there is a vertex $b\in Y_j$ that is not contained in $V(D)$ and is not incident with any edge in $S^*$. We replace $a$ by $b$ in $V(D)$ for every such $a$, and obtain an induced subgraph $D'$ whose vertices are all included in $V(G^*/S^*)$. Since $N_G(a)=N_G(b)$ for every such $a$, it is easy to see that $D'$ is isomorphic to $D$, contradicting to the fact that $G^*/S^*$ is a split graph that is $D$-free. Thus $G/S^*$ is $\{2K_2,C_4,C_5\}$-free, and $S^*$ is a solution of $G$.
\end{proof}

\vskip 0.25cm
After applying Rule 1 and Rule 2, we reduce $G$ into a graph of $O(2^{4k}k)$ vertices, implying that the problem can be solved in $2^{O(k^2)}+O(m)$ time using exhaustive search if $|K_1|\leq 2k$.

\vskip 0.25cm
Our FPT algorithm for {\sc Split Contraction} is summarized in the following table:

\begin{table}[H]
\centering
\begin{tabular}{|p{12.5cm}|}
\hline\\
\footnotesize
{\bf Algorithm {\sc Split Contraction} $(G,k)$}
\begin{itemize}
\item Find an induced split subgraph $H=(K_H;I_H)$ of size $(n-2k)$ in $G$. If it does not exist, output ``NO''. Let $V_k=V(G)-V(H)$.
\item If $|K_H|>2k$, then
	\begin{itemize}
	\item Branch into instances $(G',k')$ by contracting edges $E'\subseteq E[V_k]$. Enumerate all partitions $V'_k=(R,K_p,I_p)$.
	\item Case 1: Find $T_1=\{v\in I_H\:|\:\exists x\in I_p, vx\in E(G')\}$. Determine whether $(G'[T_1\cup K_H\cup R\cup K_p],k')$ is a yes-instance of {\sc Clique Contraction}. If yes, output ``YES''.
	\item Case 2: For every $w\in K_H$, find $T_2=\{v\in I_H\:|\:\exists x\in I_p\cup\{w\}, vx\in E(G')\}$. Determine whether $(G'[T_2\cup (K_H-\{w\})\cup R\cup K_p],k')$ is a yes-instance of {\sc Clique Contraction}. If yes, output ``YES''.
	\item If for every $(G',k')$ neither Case 1 nor Case 2 output ``YES'', then output ``NO''.
	\end{itemize}
\item Else
	\begin{itemize}
	\item Partition $V(G)$ into disjoint sets $X_1,\cdots,X_d$ such that each $X_i$ induces a maximal independent set satisfying that vertices in $X_i$ have the same neighbors.
	\item Reduction Rule 1: If $d>2^{4k}+4k$, then output ``NO''.
	\item Reduction Rule 2: If there are more than $2k+5$ vertices in $X_i$ for some $i$, then arbitrarily retain $2k+5$ vertices among them and remove others in $X_i$ from $G$.
	\item Use exhaustive search to find a solution in the reduced graph $G^*$.
	\end{itemize}
\end{itemize}\\
\hline
\end{tabular}
\end{table}

\vskip 0.25cm
\begin{theorem}\label{splitfpt}
{\sc Split Contraction} can be solved in $2^{O(k^2)}n^{O(1)}$ time.
\end{theorem}

\begin{proof}
In the above algorithm, we use an $2^{2k}n^{O(1)}$ time algorithm to find an $(n-2k)$-vertex induced split subgraph $H$. If $|K_H|>2k$, we branch into at most $|E[V_k]|^k=k^{O(k)}$ instances and enumerate at most $3^{|V'_k|}=3^{O(k)}$ partitions, and for each resulting instance it costs $k^{O(k)}+O(m)$ time to determine whether it is a yes-instance. Thus the running time of this step is bounded by $k^{O(k)}n^{O(1)}$. If $|K_H|\leq2k$, the running time of the algorithm is $2^{O(k^2)}+O(m)$. Therefore the total running time is $2^{O(k^2)}n^{O(1)}$.
\end{proof}


\section{Contracting split graphs to threshold graphs}

After obtaining an FPT algorithm for contracting general graphs to split graphs, we consider the problem of contracting split graphs to threshold graphs, which is formally defined as {\sc Threshold Contraction} on split graphs: Given a split graph $G$ and an integer $k$, can we obtain a threshold graph from $G$ by contracting at most $k$ edges?

Threshold graphs constitute a subclass of both split graphs and interval graphs in graph theory. A graph $G$ is a threshold graph if it is a split graph and for any split partition $(K;I)$ of $G$, there is an ordering $(u_1,u_2,\cdots,u_p)$ of the vertices in $K$ such that $N[u_1]\subseteq N[u_2]\subseteq\cdots\subseteq N[u_p]$ and there is an ordering $(v_1,v_2,\cdots,v_q)$ of the vertices in $I$ such that $N(v_1)\supseteq N(v_2)\supseteq\cdots\supseteq N(v_q)$, i.e., a split graph containing no induced $P_4$.

Note that the class of split graphs is closed under edge contractions, and contracting edges in a split graph never generates a new induced $P_4$ because split graphs contain no induced $P_l$ for $l>4$. Thus contracting a split graph to a threshold graph is essentially removing all induced $P_4$ from the split graph by edge contractions. Based on this observation, we obtain the following FPT algorithm.

\vskip 0.25cm
Given an input split graph $G$ with $n$ vertices and an integer $k$, we find a split partition $(K;I)$ for $G$ where $K$ is a maximal clique and $I$ is an independent set. We construct a bounded search tree whose root is labelled by $(G,k,T)$, where $T$ represents a set of vertices in $K$ that are incident with edges in a solution set and is initially an empty set $\emptyset$.

\begin{observation}
For any $v\in I$ and any $x,y\in K$ such that $vx,vy\in E(G)$, these two graphs, $G/vx$ and $G/vy$, are isomorphic. 
\end{observation}

For any induced $4$-path $v_1v_2v_3v_4$ in $G$ where $v_1,v_4\in I$ and $v_2,v_3\in K$, there are five ways to destroy this path by edge contractions: contracting $v_1x\in E(G)$ for any $x\in K$, contracting $v_4y\in E(G)$ for any $y\in K$, contracting $v_2v_3$, making $v_2$ adjacent to $v_4$, or making $v_3$ adjacent to $v_1$. The first two cases are based on Observation 4.1. The last two cases imply that $v_2$ or $v_3$ is involved in at least one edge contraction. So we branch into $5$ instances: $(G/v_1v_2,k-1,T)$, $(G/v_3v_4,k-1,T)$, $(G/v_2v_3,k-1,T)$, $(G,k,T\cup\{v_2\})$, and $(G,k,T\cup\{v_3\})$. For the fourth case, we mark every induced $P_4$ containing both $v_2$ and $v_4$ in the current graph. For the fifth case, we mark every induced $P_4$ containing both $v_3$ and $v_1$ in the current graph. We continue branching whenever $k>0$, $|T|<2k$, and the graph contains an induced $P_4$ that is not marked. Clearly the search tree has height at most $2k$, and its size is bounded by $O(5^{2k})$.

\begin{proposition}
$(G,k)$ is a yes-instance iff there exists at least one leaf $(G',k',T')$ in the search tree such that $(G',k')$ is a yes-instance, $|T'|\leq2k'$, and every induced $P_4$ in $G'$ is marked.
\end{proposition}

\vskip 0.25cm
Next, we determine whether a leaf $(G',k',T')$ is a yes-instance in FPT time. Let $(K';I')$ be the split partition of $G'$ corresponding to $(K;I)$. Since every induced $P_4$ in $G'$ is marked and thus intersects $T'$, then $G'-T'$ contains no induced $P_4$. Therefore $(K'-T';I')$ constitutes a threshold graph and there exists a vertex $u\in K'-T'$ such that $N[u]\supseteq N[w]$ for every $w\in K'-T'$. For each $w\in K'-T'$, let $P(w)$ be a set of exactly those vertices $v\in T'$ such that there exists an induced $P_4$ containing the edge $wv$ in $G'$. We group the vertices $\{w\in K'-T':|P(w)|>0\}$ into different sets $R_1,\cdots,R_d$ by $P(w)$ (i.e., any $w_1$ and $w_2$ are grouped into a same set iff $P(w_1)=P(w_2)\neq\emptyset$), and guarantee that $|R_i|\leq2k'+1$ for each $i$. If there are more than $2k'+1$ vertices $w$ that have the same $P(w)$ values, then choose $2k+1$ of them with largest degrees to form $R_i$. Obviously $d<2^{|T'|}\leq 2^{2k'}$.

Suppose that $(G',k')$ is a yes-instance and has a minimal solution set $S$. If there is a vertex $w\in I'$ that is incident with some edge $wv$ in $S$, then removing $wv$ from $S$ also yields a solution set because every induced $P_4$ in $G'$ is marked, contradicting to the fact that $S$ is minimal. Thus we have $S\subseteq E[K']$. We now use $S$ to obtain a solution set that is entirely contained in a set whose size is bounded by a function of $k'$.

We first replace every $xw\in S$ that $x\in K',w\in K'-T'$, and $w$ is not contained in any induced $P_4$ (i.e., $|P(w)|=\emptyset$) by an edge $xu$, and then obtain a set $S_1$. Note that contracting $xw$ only affects such induced $4$-path $a_1xa_2a_3$ that $a_2\in K'$ and $a_1,a_3\in I'$. To destroy this path, $w$ must be adjacent to $a_3$ in $G'$. Thus $u$ is adjacent to $a_3$ because $N[w]\subseteq N[u]$, implying that contracting $xu$ instead of $xw$ also destroys the path $a_1xa_2a_3$ (see Fig. 2(a)).

\begin{lemma}
$S_1$ is a solution set of $(G',k')$ if $S$ is.
\end{lemma}

\begin{figure}[H]
\centering
\begin{tikzpicture}
  \tikzset{every node/.style={circle,draw,inner sep=0pt,minimum size=5pt}}
  \node[fill=black] at (1,3)(a1){};
  \node[fill=black] at (1,2)(x){};
  \node[fill=black] at (3,2)(a2){};
  \node[fill=black] at (3,3)(a3){};
  \node[fill=black] at (1.6,1.5)(u){};
  \node[fill=black] at (2.4,1.5)(w){};
  \draw[rounded corners,dashed](0.5,2.35) rectangle (3.5,1);
  \draw[red](a1)--(x)--(a2)--(a3);
  \draw(x)--(w)--(a3)--(u)--(x);
  \bubble(1.6,1.25){$u$};
  \bubble(2.4,1.25){$w$};
  \bubble(1,1.75){$x$};
  \bubble(1,3.25){$a_1$};
  \bubble(3,1.75){$a_2$};
  \bubble(3,3.25){$a_3$};
  \bubble(4,1.5){$K'$};
  \bubble(2,0.5){(a)};

  \node[fill=black] at(7.75,3)(b1){};
  \node[fill=black] at(7,3)(b4){};
  \node[fill=black] at(9,3)(b3){};
  \node[fill=black] at(9.5,3)(b5){};
  \node[fill=black] at(9,1.5)(b2){};
  \node[fill=black] at(7.75,2)(w1){};
  \node[fill=black] at(7,1.5)(wx){};
  \draw[rounded corners,dashed](6.5,2.35) rectangle (9.5,1);
  \draw[red](b1)--(w1)--(b2)--(b3);
  \draw[blue](b4)--(wx)--(b2)--(b5);
  \draw(wx)--(b1);
  \draw[green,dashed](b1)--(b2);
  \bubble(7,1.25){$w'$};
  \bubble(7.75,1.75){$w$};
  \bubble(9,1.25){$b_2$};
  \bubble(7,3.25){$b_4$};
  \bubble(7.75,3.25){$b_1$};
  \bubble(9,3.25){$b_3$};
  \bubble(9.5,3.25){$b_5$};
  \bubble(10,1.5){$K'$};
  \bubble(8,0.5){(b)};

\end{tikzpicture}
\caption{(a) $a_1xa_2a_3$ can be destroyed by contracting $xu$ instead of $xw$; (b) $b_1wb_2b_3$ can be destroyed by contracting $S_1-\{wy\}$}
\end{figure}
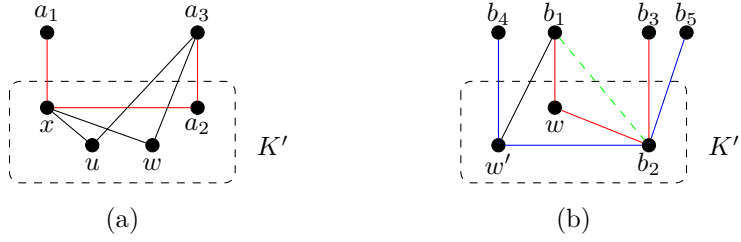

Furthermore for any $w\in K'-T'$ such that $P(w)\neq\emptyset$ and $w$ is not contained in any $R_i$, there exits a vertex $w'\in R_j$ for some $j$ such that $P(w')=P(w)$ and $w'$ is not involved in any edge contraction since $|R_j|=2k'+1>2k'$. We replace the edges $\{wy\in S_1: y\in K'\}$ incident with $w$ by $\{w'y: wy\in S_1\}$ for every such $w$ and then obtain a set $S_2$. We show that $S_2$ is also a solution set.

\begin{lemma}
$S_2$ is a solution set of $(G',k')$ if $S_1$ is.
\end{lemma}

\begin{proof}
Note that contraction of each $wy\in S_1$ only destroys the induced $4$-paths containing either $w$ or $y$, then other induced $4$-paths in $G'$ are destroyed by contracting $S_1-\{wy\}$. Since $w'\in R_j,w\notin R_j$, and $P(w')=P(w)$, we have $N[w]\subseteq N[w']$ by the construction rule of $R_j$. Similar to the above proof for $S_1$, contracting $w'y$ instead of $wy$ will destroy the paths that contain $y$ but not $w$.

Moreover for any induced path $b_1wb_2b_3$ containing $w$, there exists another induced path $b_4w'b_2b_5$ containing $w'$ and $b_2$ in $G'$ because $b_2\in P(w)=P(w')$. Clearly $\{b_1,w',b_2,b_5\}$ constitute an induced $P_4$ in $G'$ since $b_1$ and $b_2$ are not adjacent and $w'$ and $b_5$ are not adjacent either. Note that $b_1w'b_2b_5$ is destroyed by contracting $S_1-\{wy\}$ and $w'$ is not involved in any edge contraction. Therefore $b_2$ must be adjacent to $b_1$ after contractions, implying that $b_1wb_2b_3$ can be destroyed by contracting $S_1-\{wy\}$ (see Fig. 2(b)).

Thus contracting $(S_1-\{wy\})\cup\{w'y\}$ also removes all induced $4$-paths in $G'$, implying that $S_2$ is a solution.
\end{proof}

Let $R=R_1\cup\cdots\cup R_d$, we have the following result.

\begin{corollary}
$(G',k')$ has a solution set that is entirely contained in $E[T'\cup R\cup\{u\}]$ if it is a yes-instance.
\end{corollary}

Thus to determine whether $(G',k')$ is a yes-instance, we use exhaustive search to check whether there is a $k'$-size solution set in $E[T'\cup R\cup\{u\}]$, which costs $O(\binom{|T'|+|R|+1}{2}^{k'})=2^{O(k'^2)}$ time since $|R|\leq (2k'+1)2^{2k'}$. We have obtained an FPT algorithm for {\sc Threshold Contraction} on split graphs.

\begin{theorem}\label{thresholdfpt}
{\sc Threshold Contraction} can be solved in $2^{O(k^2)}n^{O(1)}$ time for split graphs.
\end{theorem}


\section{Incompressibility}

In this section, we show that {\sc Split Contraction} and {\sc Threshold Contraction} on split graphs are very unlikely to have polynomial kernels. To this end, we give polynomial parameter transformations from two domination problems: {\sc Red-Blue Dominating Set} and {\sc One-Sided Dominating Set}. A \emph{polynomial parameter transformation} from a problem $\mathcal{Q}$ to another problem $\mathcal{Q'}$ is a polynomial-time computable function that maps $(I,k)$ to $(I',k')$ such that: (a) $(I,k)\in\mathcal{Q}\Leftrightarrow (I',k')\in\mathcal{Q'}$; (b) $k'\leq h(k)$ for some computable function $h$. Bodlaender \emph{et al.}~\cite{bodlaender08} proved that if $\mathcal{Q}$ is incompressible and there is a polynomial parameter transformation from $\mathcal{Q}$ to $\mathcal{Q'}$, then $\mathcal{Q'}$ is incompressible as well. {\sc Red-Blue Dominating Set} and {\sc One-Sided Dominating Set} are defined as follows:

\vskip 0.25cm
\begin{tabular}{l l}
\multicolumn{2}{l}{\sc Red-Blue Dominating Set}\\
{\it Instance}: & Bipartite graph $G=(X,Y;E)$ and an integer $t$.\\
{\it Question}: & Does $Y$ have a subset of at most $t$ vertices that dominates $X$?\\
{\it Parameter}: & $|X|,t$.\\
\end{tabular}
\vskip 0.25cm

\vskip 0.25cm
\begin{tabular}{l l}
\multicolumn{2}{l}{\sc One-Sided Dominating Set}\\
{\it Instance}: & Bipartite graph $G=(X,Y;E)$ and an integer $t$.\\
{\it Question}: & Does $X$ have a subset of at most $t$ vertices that dominates $Y$?\\
{\it Parameter}: & $|X|$.\\
\end{tabular}
\vskip 0.25cm

These two problems of similar definitions have different meanings. If we consider the set $X$ of the input bipartite graph as the ``small side'' since the cardinality of $X$ is the parameter of the problem, and the set $Y$ as the ``large side'' since the cardinality of $Y$ could be very large, the {\sc Red-Blue Dominating Set} problem is to find a set in ``large side'' to dominate ``small side'' while the {\sc One-Sided Dominating Set} problem is to find a set in ``small side'' to dominate ``large side''.

Dom \emph{et al.}~\cite{dom09} proved that {\sc Red-Blue Dominating Set} is incompressible using the Colors and IDs technique. This result is applicable to several other incompressible problems such as {\sc Steiner Tree}, {Connected Vertex Cover}~\cite{dom09}, and {\sc Tree Contraction}~\cite{heggernes11c}.

\vskip 0.25cm
\begin{theorem}[Dom \emph{et al.}\cite{dom09}]\label{rbdsincompress}
{\sc Red-Blue Dominating Set} admits no polynomial kernel unless $NP\subseteq coNP/poly$.
\end{theorem}

We point out that {\sc One-Sided Dominating Set} can be solved in $O(2^{|X|}|I|^{O(1)})$ time by enumerating all subsets of $X$, and is incompressible. The proof will appear in full version of the paper~\cite{cai13} by the same authors.

\vskip 0.25cm
\begin{theorem}\label{osdsincompress}
{\sc One-Sided Dominating Set} admits no polynomial kernel unless $NP\subseteq coNP/poly$.
\end{theorem}

\vskip 0.25cm
We now give a polynomial parameter transformation from {\sc Red-Blue Dominating Set} to {\sc Split Contraction}. Without loss of generality, we may assume that every vertex in $X$ has at least one neighbor in $Y$ and then $t\leq|X|$. Our result is inspired by the reduction for {\sc Tree Contraction} by Heggernes \emph{et al.}~\cite{heggernes11c}. 

\begin{theorem}\label{splitincompress}
{\sc Split Contraction} admits no polynomial kernel unless $NP\subseteq coNP/poly$.
\end{theorem}

\begin{proof}
Given a bipartite graph $G=(X,Y;E)$ and a positive integers $t$, we construct a graph $G'$ from $G$ by creating a clique $C$ of size $|X|+t+3$ where a designated vertex $u\in C$ is made adjacent to all vertices of $Y$, and for every $v\in X$ adding $|X|+t+1$ new leaves appending to $v$. We claim that $Y$ has a subset of at most $t$ vertices that dominates $X$ iff $G'$ can be made into a split graph by contracting at most $|B|+t$ edges.

Suppose that $Y'$ is a $t$-subset of $Y$ that dominates $X$. Since vertices of $\{u\}\cup Y' \cup X$ induce a connected graph, we can make these $|X|+t+1$ vertices into a single vertex by using $|X|+t$ edge contractions. The subgraph of $G'$ induced by $\{u\}\cup Y \cup X$ is contracted to a star, and $G'$ is modified into a split graph after contractions.

Conversely, suppose that $G'$ contains at most $|X|+t$ edges $F$ whose contraction results in a split graph. Note that at least two vertices $a$ and $b$ other than $u$ in $C$ survive after contracting $F$, and for each vertex $x\in X$ there exists a leaf $x'$ appending to $x$ that is not involved in any edge contraction. If any $x\in X$ and $u$ are in the different witness set, then $\{x,x',a,b\}$ form an induced $2K_2$, contradicting to the fact that $G'/F$ is a split graph. Therefore all vertices in $X\cup\{u\}$ are in the same witness set. Observe that each path starting from one vertex in $X$ to $u$ must go through some vertices in $Y$, which implies that $X$ is dominated by a subset $I$ of $Y$ containing exactly those vertices that are in the same witness set with $u$. Thus we obtain a solution of $G$ with $|I|\leq|F|-|X|\leq t$ vertices.
\end{proof}

In order to show the incompressibility of {\sc Threshold Contraction} in a clear sketch, we first give a polynomial parameter transformation from {\sc One-Sided Dominating Set} to the following {\sc One-Sided Domatic Number} problem, and next give another transformation from {\sc One-Sided Domatic Number} to {\sc Threshold Contraction} on split graphs.

\vskip 0.25cm
\begin{tabular}{l l}
\multicolumn{2}{l}{\sc One-Sided Domatic Number}\\
{\it Instance}: & Bipartite graph $G=(X,Y;E)$ and integer $t$.\\
{\it Question}: & Can we partition $X$ into $t$ disjoint sets such that vertices in each\\
 & set dominate $Y$?\\
{\it Parameter}: & $|X|$.\\
\end{tabular}
\vskip 0.25cm

{\sc One-Sided Domatic Number} can be solved in $O(2^{|X|log|X|}|I|^{O(1)})$ time by enumerating all partitions for $X$, and is incompressible.

\vskip 0.25cm
\begin{theorem}\label{osdnincompress}
{\sc One-Sided Domatic Number} admits no polynomial kernel unless $NP\subseteq coNP/poly$.
\end{theorem}

\begin{proof}
Given a bipartite graph $G=(X,Y;E)$ and an integer $t$, we construct a graph $G'$ by creating a new vertex $w$ that is made adjacent to every vertex of $X$, and a set $Z$ of $|X|-t$ vertices that are made adjacent to every vertex of $Y$.

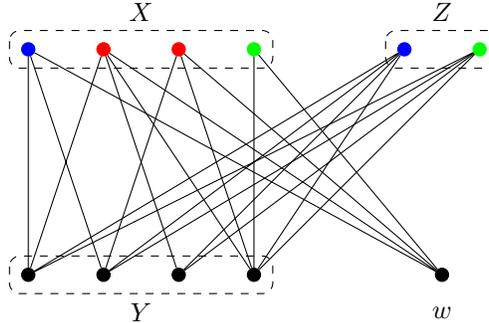
\begin{figure}[H]
\centering
\begin{tikzpicture}
  \tikzset{every node/.style={circle,draw,inner sep=0pt,minimum size=5pt}}
  \node[blue,fill=blue] at (1,4)(a1){};
  \node[red,fill=red] at (2,4)(b1){};
  \node[red,fill=red] at (3,4)(c1){};
  \node[green,fill=green] at (4,4)(d1){};
  \node[fill=black] at (1,1)(a2){};
  \node[fill=black] at (2,1)(b2){};
  \node[fill=black] at (3,1)(c2){};
  \node[fill=black] at (4,1)(d2){};
  \draw (b1)--(a2)--(a1)--(b2)--(c1)--(d2);
  \draw (c2)--(b1)--(d2)--(d1);
  \draw[rounded corners,dashed] (0.75,4.25) rectangle (4.25,3.75);
  \draw[rounded corners,dashed] (0.75,1.25) rectangle (4.25,0.75);
  \bubble(2.5,4.5){$X$};
  \bubble(2.5,0.5){$Y$};

  \node[blue,fill=blue] at (6,4)(g1){};
  \node[green,fill=green] at (7,4)(g2){};
  \node[fill=black] at (6.5,1)(h){};
  \draw (a2)--(g1)--(b2)--(g2)--(a2); \draw (c2)--(g1)--(d2)--(g2)--(c2);
  \draw (a1)--(h)--(b1); \draw (c1)--(h)--(d1);
  \draw[rounded corners,dashed] (5.75,4.25) rectangle (7.25,3.75);
  \bubble(6.5,4.5){$Z$};
  \bubble(6.5,0.5){$w$};
\end{tikzpicture}
\caption{An example of construction of $G'$ with $|X|=4,t=2$.}
\end{figure}

We claim that $X$ has a subset of at most $t$ vertices that dominates $Y$ iff $X\cup Z$ can be partitioned into $|X|-t+1$ disjoint sets such that vertices in each set dominate $Y\cup\{w\}$.

Suppose that $X$ has a subset $S$ of $t$ vertices that dominates $Y$. Obviously $S$ dominates $Y\cup\{w\}$. Note that every vertex in $X\setminus S$ is adjacent to $w$ and every vertex in $Z$ dominates $Y$. Thus vertices in $(X\setminus S)\cup Z$ can be partitioned into $|X|-t$ pairs such that each pair of vertices dominates $Y\cup\{w\}$. Included the set $S$, we obtain $|X|-t+1$ disjoint dominating sets for $Y\cup\{w\}$ in $X\cup Z$.

Conversely, suppose that $X\cup Z$ has a disjoint partition $(S_1,\cdots, S_{|X|-t+1})$ such that each $S_i$ dominates $Y\cup\{w\}$. Note that each $S_i$ that contains some vertex of $Z$ must contain at least one vertex in $X$ because $w$ is only adjacent to vertices in $X$. Let $r$ be the number of sets $S_i$ that intersect $Z$ where $r\leq |Z|=|X|-t$. It is easy to see that there are at most $|X|-r$ vertices in $X$ that constitute $|X|-t+1-r$ disjoint dominating sets for $Y\cup\{w\}$, and thus there exists one dominating set containing at most $\frac{|X|-r}{|X|-t+1-r}=1+\frac{t-1}{|X|-t+1-r}\leq t$ vertices which is clearly a dominating set for $Y$.

Since $|X\cup Z|=2|X|-t$, the reduction is parameter preserving.
\end{proof}
\vskip 0.25cm

\begin{theorem}\label{thresholdincompress}
{\sc Threshold Contraction} on split graphs admits no polynomial kernel unless $NP\subseteq coNP/poly$.
\end{theorem}

\begin{proof}
Given a bipartite graph $G=(X,Y;E)$ where $X=\{x_1,\cdots,x_k\}$ and $Y=\{y_1,\cdots,y_n\}$, and a positive integer $t$, we construct a graph $G'$ as following:

\begin{itemize}
\item Create a clique $K$ of $|X|$ vertices $\{u_1,\cdots,u_k\}$, and a clique $A$ of size $2|X|+1$. Make $K\cup A$ into a large clique.
\item Create an independent set $B$ of $|X|+1$ vertices that are made adjacent to every vertex of $K$, and create $|X|+1$ disjoint independent sets of size $|Y|$: $I_l=\{v_1^{(l)},\cdots,v_n^{(l)}\}$ ($l=1,\cdots,|X|+1$). All vertices in $I_1\cup\cdots\cup I_{|X|+1}$ are made adjacent to every vertex of $A$.
\item For each $1\leq l\leq |X|+1$, make $u_i$ adjacent to $v_j^{(l)}$ iff $x_i y_j\in E$.
\item $G'=(K\cup A;B\cup I_1\cup\cdots\cup I_{|X|+1})$ is a split graph.
\end{itemize}

\begin{figure}[H]
\centering
\begin{tikzpicture}
  \tikzset{every node/.style={circle,draw,inner sep=0pt,minimum size=5pt}}
  \node[fill=black] at (0,3)(r){};
  \node[fill=black] at (0.75,3)(s){};
  \node[fill=black] at (1.5,3)(t){};
  \node[fill=black] at (0.4,1.5)(r1){};
  \node[fill=black] at (1.1,1.5)(s1){};
  \draw (r)--(r1)--(s)--(s1)--(t);
  \bubble(0.75,3.5){$X$};
  \bubble(0.75,1){$Y$};

  \node[fill=black] at (6,3.5)(k1){};
  \node[fill=black] at (7,3.75)(k2){};
  \node[fill=black] at (8,3.5)(k3){};
  \draw[rounded corners,dashed] (5.75,4) rectangle (8.25,3.25);
  \node[rectangle,rounded corners,minimum width=1cm,minimum height=0.5cm] at (9.5,3.75)(x){$A$};
  \draw (k1)--(k2)--(k3)--(x)--(k1)--(k3); \draw (x)--(k2);
  \node[rectangle,rounded corners,minimum width=1cm,minimum height=0.5cm] at (5,1)(y){$B$};
  \draw (k1)--(y)--(k2); \draw (y)--(k3);
  \node[fill=black] at (6.5,1)(i11){};
  \node[fill=black] at (7.5,1)(i12){};
  \draw[rounded corners,dashed] (6.25,1.25) rectangle (7.75,0.75);
  \draw (k1)--(i11)--(k2)--(i12)--(k3);
  \node[fill=black] at (9.5,1)(in1){};
  \node[fill=black] at (10.5,1)(in2){};
  \draw[rounded corners,dashed] (9.25,1.25) rectangle (10.75,0.75);
  \draw (k1)--(in1)--(k2)--(in2)--(k3);
  \draw (i11)--(x);
  \draw (i12)--(x);
  \draw (in1)--(x);
  \draw (in2)--(x);

  \bubble(8.5,1){$\cdots$};
  \bubble(7,4.25){$K$};
  \bubble(7.25,0.5){$I_1$};
  \bubble(10,0.5){$I_{k+1}$};
  \draw[line width=5pt,gray,-latex](2.5,2.25)--(4,2.25);
\end{tikzpicture}
\caption{An example of the reduction from {\sc Red-Blue Domatic Number} to {\sc Threshold Contraction}.}
\end{figure}
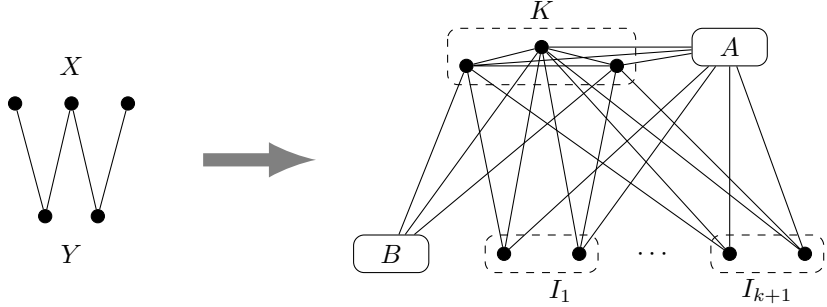

We claim that $X$ can be partitioned into $t$ disjoint dominating sets for $Y$ iff $G'$ can be modified into a threshold graph by contracting at most $|X|-t$ edges.

Suppose that the set $B$ has a disjoint partition $X=X_1\cup\cdots\cup X_t$ such that each $X_i$ dominates $Y$. We partition the set $\{u_1,\cdots,u_k\}$ into $t$ disjoint sets $\{R_1,\cdots,R_t\}$ corresponding to $\{X_1,\cdots,X_t\}$, and contract vertices in each $R_i$ to a single vertex $r_i$. The total number of edge contractions we use is $\Sigma_i (|R_i|-1)=(\Sigma_i |R_i|)-t=|X|-t$, and $r_i$ is adjacent to every vertex of $I_1\cup\cdots\cup I_{|X|+1}$ for each $1\leq i\leq t$ because $X_i$ dominates $Y$ in $G$. The closed neighbor sets of vertices in $K\cup A$ satisfy that $N[r_1]=\cdots=N[r_t]\supseteq N[A]$ after contractions, thus $G'$ is made into a threshold graph.

Conversely, suppose that $G'$ contains at most $|X|-t$ edges $F$ whose contraction results in a threshold graph. $K$ is partitioned into $r$ witness sets $W_1,\cdots,W_r$, and vertices in each $W_j$ are contracted to a single vertex $w_j$. Note that $\Sigma_j(|W_j|-1)\leq|F|$, we have $r\geq \Sigma_j|W_j|-|F|\geq t$. Since $|F|\leq |X|$, there exists some $1\leq p\leq |X|+1$ such that none vertex in $I_p$ is incident with any edge in $F$. Moreover, there exist vertices $a\in A$ and $b\in B$ that are not incident with any edge in $F$. By the construction, each vertex in $K$ is adjacent to $b$ which is non-adjacent to $a$, and $a$ is adjacent to all vertices in $I_p$. Thus to modify $G'$ into a threshold graph, each vertex in $\{w_1,\cdots,w_r\}$ must be adjacent to all vertices in $I_p=\{v_{p,1},\cdots,v_{p,n}\}$ after contractions, implying that $W_j$ dominates $\{y_i,\cdots,y_n\}$ in $G$ for each $1\leq j\leq r$. Therefore $(G,t)$ is a yes-instance of {\sc One-Sided Domatic Number}.
\end{proof}
\vskip 0.25cm

The above reduction implies the NP-hardness of {\sc Threshold Contraction}.

\begin{theorem}\label{thresholdnpc}
{\sc Threshold Contraction} is NP-complete even for split graphs.
\end{theorem}


\section{Concluding Remarks}

In this paper we have proved that the following two edge contraction problems are FPT but admit no polynomial kernels unless $NP\subseteq coNP/poly$: {\sc Split Contraction}, and {\sc Threshold Contraction} on split graphs. We note that {\sc Threshold Contraction} on split graphs is essentially the same as {\sc Cograph Contraction} on split graphs. The {\sc Cograph Contraction} problem is claimed to be FPT for general graphs by a recent paper (Lokshtanov \emph{et al.}~\cite{lokshtanov13}) using the notion of rankwidth. We have provided a conceptually simpler and faster algorithm for the problem when the input is a split graph. Our result in Theorem~\ref{thresholdincompress} also implies that {\sc Cograph Contraction} admits no polynomial kernel unless $NP\subseteq coNP/poly$.

Although we have proved that contracting general graphs to split graphs and contracting split graphs to threshold graphs are both $FPT$, our results do not imply the fixed-parameter tractability for {\sc Threshold Contraction} on general graphs since our algorithm for {\sc Split Contraction} does not find all split graphs obtained by contractions. It remains open whether {\sc Threshold Contraction} is FPT for general graphs.



\end{document}